\documentclass[11pt,a4]{article}
\usepackage{asaf}
 
\title{Fast Metric Decompositions in High Dimension%
  \thanks{This paper will appear in the Proceedings of ESA 2026. 
  This version includes an additional note before \Cref{sec:padded-decom}. }
} 

\author{Robert Krauthgamer%
  \thanks{The Harry Weinrebe Professorial Chair of Computer Science.
    This research was supported by the Israel Science Foundation grant \#1336/23.
  Email: \texttt{robert.krauthgamer@weizmann.ac.il}
  } 
  \qquad
  Asaf Petruschka%
  \thanks{
  Email: \texttt{asaf.petruschka@weizmann.ac.il}
  } 
  \qquad 
  Nir Petruschka\thanks{
    Email: \texttt{nir.petruschka@weizmann.ac.il}
  } 
\\  Weizmann Institute of Science
}

\begin{document}

\maketitle

\begin{abstract}
Metric decompositions are a fundamental tool in the design of algorithms involving distances. 
We study fast algorithms for sampling from probabilistic metric decompositions of $n$-point sets in $\ell_\infty$ and $\ell_2$ spaces of high dimension $d$.
For $\ell_\infty$, we design a padded-decomposition algorithm 
that runs in time $\tilde{O}(nd^2)$, which is near-linear in $n$,
and achieves padding parameter $\tilde{O}(\log n)$. 
Our algorithm constructs a new sparse neighborhood cover
that is based on geometric properties of $\ell_\infty$ [Indyk, JCSS'01],
and utilizes recent reductions between covers and decompositions [Conroy and Filtser, STOC'25]. 
For $\ell_2$, we design a separating-decomposition algorithm 
that achieves near optimal separation $\tilde{O}(\sqrt{\log n})$ in almost-linear time $n^{1+o(1)}$. 
Our bounds improve over known algorithms with similar running time by a factor $\Omega(\sqrt{\log n})$, and the techniques have additional applications to spanners and nearest-neighbor search.
\end{abstract}

\section{Introduction}

Metric decompositions are fundamental algorithmic primitives when dealing with data that involves distances.
Informally, a decomposition refers to a probabilistic partition of the points
into low-diameter \emph{clusters},
such that close-by points ``tend'' to be clustered together.
Decompositions are a very useful tool in divide-and-conquer algorithms,
which first handle each low-diameter cluster separately,
and then combine the solutions into a global one.
These decompositions have numerous applications, 
including metric embeddings \cite{Bartal96,Rao99,GKL03,KLMN05,LN05,FL21}, 
approximation algorithms for cut problems \cite{LLR95,GVY96,FHL08},
vertex sparsification \cite{CLLM10,EGKRTT14,MM16, KKN15,Filtser19b}, 
spanners, distance oracles and routing \cite{AGGM06,MN07,ENW16,HIS13,FN22,HMO21}, 
parallel algorithms for linear algebra \cite{BGKMPT14, MPX13}, 
fast algorithms for single-source shortest paths \cite{BNW25, LMR26},
and spectral graph theory \cite{KLPT09, BLR10}. 

Modern applications are scaling rapidly, 
and current datasets are often massive and high-dimensional.
As exceeding linear time significantly is simply prohibitive,
this creates a strong demand for \emph{fast algorithms}. 
Prior work on fast algorithms for metric decompositions has addressed graph metrics,
whereas we study the geometric setting of points in $\RR^d$. 
This line of work aims to match the optimal quality bounds 
but using algorithms whose running time is close to linear (in the input size). 

We consider two classical notions of probabilistic metric decompositions: \emph{separating decompositions} (also known as \emph{Lipschitz decompositions}) and \emph{padded decompositions}, defined as follows. 
Let $(X,d_X)$ be a metric space.
For a partition $\cP$ of $X$,
we call each element of $\cP$ (which is a subset of $X$) a \emph{cluster},
and we denote by $\cP(x)$ the cluster containing point $x\in X$.
The partition $\cP$ is called \emph{$\Delta$-bounded}
if every cluster in it has diameter at most $\Delta > 0$,
i.e., the distance between every two points in the same cluster is at most $\Delta$.

\begin{definition}[Separating Decomposition~\cite{Bartal96}]
  A distribution $\cD$ over $\Delta$-bounded partitions of $X$ is called
  an \emph{$(\alpha, \Delta)$-separating decomposition} if
    \[
    \forall x,y \in X, \qquad
    \Pr_{\cP\sim \cD} [\cP(x) \neq \cP(y)] \leq \alpha \cdot \frac{d_X (x,y)}{\Delta}.
    \]
\end{definition}

\begin{definition}[Padded Decomposition~\cite{Rao99,KL07,Filtser19}]
  A distribution $\cD$ over $\Delta$-bounded partitions of $X$ is called
  a \emph{$(\beta, \delta, \Delta)$-padded decomposition} if 
    \[
    \forall x \in X, \gamma \in [0,\delta], \qquad
    \Pr_{\cP\sim \cD} [B(x,\gamma \Delta) \subseteq \cP(x)] \geq e^{-\beta\gamma} ,
    \]
    where $B(x,r) := \{y \in X \mid d_X (x,y) \leq r\}$ is the closed ball of radius $r\ge0$ around $x$.
\end{definition}
We note that some literature uses variants of these definitions;
for example, it is common to define padded decomposition using only the case $\gamma = \frac{\ln 2}{\beta}$, 
and then $\delta$ is not needed.
Under our definitions above, every $(\beta,\delta,\Delta)$-padded decomposition is also a $(\max\{\beta,\delta^{-1}\}, \Delta)$-separating decomposition.
To see this, consider $x,y\in X$;
  if $d_X(x,y) \geq \delta\Delta$
  then $\Pr[\cP(x) \neq \cP(y)] \leq 1 \leq \frac{d_X(x,y)}{\delta\Delta}$;
  otherwise, $\Pr[\cP(x) \neq \cP(y)] \leq \Pr[B(x, d_X(x,y)) \not\subseteq \cP(x)] \leq 1-e^{-\beta\frac{d_X(x,y)}{\Delta}} \leq \beta\frac{d_X(x,y)}{\Delta}$.  
An earlier connection established in \cite[Theorem 2.2]{LN04pub} 
shows that even a $(\beta,\delta,\Delta)$-padded decomposition  
only for the case $\gamma = \frac{\ln 2}{\beta}$
implies a $(4\beta,2\Delta)$-separating decomposition for the same metric.

We study finite datasets $X$ of size $n=|X|$
that lie in $\RR^d$ equipped with the $\ell_\infty$ or $\ell_2$ norm
(which we denote by $\ell_\infty^d$ and $\ell_2^d$).
We focus on the high-dimensional regime, namely $d \geq \omega(\log n)$,
in which case running time that is exponential in $d$ is prohibitive, 
as it is super-polynomial in $n$.
For ease of presentation, we use the following terminology.
A \emph{padded-decomposition algorithm of quality $q\ge 1$} is an algorithm that,
given a metric space $(X,d_X)$ and a diameter bound $\Delta > 0$,
outputs a partition sampled from a $(\beta,\delta,\Delta)$-padded decomposition
with $\max\{\beta,\delta^{-1}\} \leq q$.
The analogous definition regarding $(\alpha, \Delta)$-separating decomposition
requires $\alpha \leq q$. 
Throughout, the notation $\tilde{O}(f)$ hides a polylogarithmic factor in $f$. 
The term \emph{near-linear} in $f$ refers to $\tO(f)$, whereas \emph{almost-linear} refers to $f^{1+o(1)}$;
they usually refer to running time, in our case mostly with respect to $n$, with polynomial dependence on $d$.

\paragraph{The spanner approach.}
Fast algorithms for metric decompositions have already been devised for \emph{graphs}.
Specifically, for input graphs $G$ with $m$ weighted edges, 
there are $\tilde{O}(m)$-time algorithms for padded (and separating) decomposition 
of optimal quality $O(\log n)$~\cite{MS09,ABCP98,MPX13,CF25}.
Thus, a natural approach is to approximate $(X,d_X)$ 
by a \emph{spanner} $G_X$, which is a sparse graph on the vertex set $X$,
and find a decomposition of this graph $G_X$.
State-of-the-art constructions of geometric spanners~\cite{HIS13,FN22,AZ23,KP25,KPS25,NR26} 
yield surprisingly strong results.
While this approach serves as a good \emph{baseline}, it has some limitations.
First, the intermediate step (constructing a spanner) may be overly restrictive 
and thus preclude achieving a tight quality bound.
Specifically, one factor going into the overall quality of the decomposition 
is the spanner's stretch, which is often $\log^c n$ for some fixed $c>0$. 
Second, by passing to a graph representation, we effectively “forget” the geometric properties of $(X,d_X)$ 
and instead treat it as an arbitrary $n$-point metric,
for which the optimal quality is $O(\log n)$ even for sparse graphs~\cite{Bartal96},
and this factor propagates into the overall quality of the decomposition.
To illustrate this issue, observe that for $X \subseteq \ell_2^d$,
the spanner approach can achieve at best quality $O(\log n)$,
even though a separating decomposition of quality $O(\sqrt{\log n})$ is known to exist~\cite{CCGGP98}. 
We seek to circumvent these limitations by leveraging the dataset's geometry directly.

\paragraph{The shifted-grid approach.}
Another standard approach for fast construction of separating decompositions
in $\ell_p^d$ spaces is to partition the dataset using a randomly shifted grid.
This approach can be implemented in $O(dn)$ time and yields a decomposition of
quality $O(d)$, which is suboptimal for the high-dimensional regime $d=\omega(\log n)$.
This approach can nevertheless be used in $\ell_2^d$, even when $d=\omega(\log n)$,
because the dimension can be reduced to $O(\log n)$ by the JL lemma~\cite{JL84}. 
A technical issue that arises is that the JL map has a small failure probability, 
which might affect the correctness of the procedure (e.g., for very close-by points).
We resolve it by verifying in near-linear-time that the decomposition has succeeded; 
see the proof of \cref{thm:fast-sep-decom} for details.
Perhaps a lesser-known fact is that the shifted-grid approach 
can be used in $\ell_1^d$ even when $d=\omega(\log n)$, 
despite the impossibility of dimension reduction in this space~\cite{BC05, LN04}. 
To construct an $(O(\log n),\Delta)$-separating decomposition, 
set the grid sidelength to $\frac{\Delta}{3\log n}$. 
For every pair of data points $x,y$, 
the probability they are separated by the grid is at most $O(\log n)\frac{\|x-y\|_1}{\Delta}$, 
and the probability they are clustered together is at most $n^{-\frac{3 \|x-y\|_1}{\Delta}}$.
The crux is that even though the $\ell_1$-diameter of a grid cell is larger than $\Delta$, 
the diameter of a cluster is actually determined by the data points it contains;
each pair of data points whose distance exceeds $\Delta$  
is clustered together with probability at most $\frac{1}{n^3}$, 
and thus by a union bound over all pairs, with high probability, no such pair is clustered together. 

We proceed to describe our results, which are also summarized in~\Cref{table:fast-decom}
in comparison to known or baseline bounds.

\begin{table}[]
\centering
\begin{tabular}{lllll}
\midrule
Metric                           & Decomposition & Time & Quality                              & Reference                                                                    \\ \midrule \midrule
general                            & padded  & polynomial & $O(\log n)$                                          & \cite{Bartal96}                                         \\ \midrule

graph                            & padded  & near-linear & $O(\log n)$                                          & \cite{MS09}                                         \\ \midrule
\multirow{4}{*}{$\ell_\infty^d$} & padded   & polynomial & $O(\log n)$                                          & \cite{Bartal96}                                         \\
                                 & padded & almost-linear & $O( \log n \cdot \log^4 d)$          & via spanners \cite{KPS25} + \cite{MS09}         \\
                                 & padded & near-linear & $O(\log^{3/2} n \cdot \log^4 d)$          & via spanners \cite{KPS25} + \cite{MS09}         \\
                                 & padded & near-linear      & $O(\log n \cdot \log \log d)$            & \Cref{thm:fast-padded-decom}                              \\
 \midrule
\multirow{3}{*}{$\ell_2^d$}      & separating & polynomial  & $O(\sqrt{\log n})$                   & \cite{CCGGP98}         \\
                                 & separating & almost-linear           & $\tO(\sqrt{\log n})$           & \Cref{thm:fast-sep-decom}                                 \\ 
                                 & separating & near-linear           & $O(\log n)$                               & via randomly shifted grid                                             \\
                                 \midrule
\end{tabular}
\caption{Comparison of metric decompositions for $n$-point metrics,
  focusing on the quality achieved in time that is polynomial vs.\ almost-linear vs.\ near-linear in $n$ 
  (ignoring the dependence on $d$ which is polynomial).
  }
\label{table:fast-decom}
\end{table}

\paragraph{$\ell_\infty$ metrics.}
Our main result is a near-linear-time padded-decomposition algorithm for $\ell_\infty$ 
with near-optimal quality $\tO(\log n)$. 
We prove the following in~\Cref{sec:padded-decom}. 

\begin{theorem}\label{thm:fast-padded-decom}
    There is a padded-decomposition algorithm for $n$-point metrics in $\ell_\infty^d$, 
    that achieves quality $O(\log n \cdot \log \log d)$ in running time $\tilde{O}(nd^2)$.
\end{theorem}

Note that by a classical result of Matousek~\cite{Matousek96}, 
every $n$-point metric embeds into $\ell_\infty^{n^{\epsilon}}$ with distortion $O(1/\epsilon)$,%
\footnote{A metric space $(X, d_X)$ embeds into a metric space $(Y, d_Y)$ with \emph{distortion} $k \geq 1$ 
if there exist a function $f: X \to Y$ and $s > 0$ 
such that for all $x, y \in X$, $\frac{s}{k} \cdot d_X(x, y) \leq d_Y(f(x), f(y)) \leq s \cdot d_X(x, y)$.
} 
and this implies, by a result of \cite{Bartal96},
that some $n$-point metrics in $\ell_\infty^{n^{\epsilon}}$ require quality $\Omega(\epsilon \log n)$. 
For fixed $\epsilon>0$, our \Cref{thm:fast-padded-decom} nearly matches this lower bound 
(with running time $\tilde{O}(n^{1+2\epsilon})$).
For comparison with the baseline, 
current $o(n^2)$-time spanner constructions \cite{KPS25}%
\footnote{We note that~\cite{NR26} show the existence of spanners in $\ell_\infty$ with size $\tilde{O}(n)$ and stretch $O(\sqrt{\log n} \log^2 d)$, but their running time is not analyzed.} 
have stretch roughly $O(\log^4 d)=O(\log^4 n)$,
and thus a baseline padded-decomposition algorithm only achieves a far worse quality $O(\log^5 n)$.

The algorithm of~\Cref{thm:fast-padded-decom} goes through the closely-related notion of 
\emph{sparse neighborhood cover}~\cite{AP90, ABCP98} (for formal definition see~\Cref{def:nbr-cover}), 
which was recently shown to yield a padded decomposition~\cite{CF25}, and is also of independent interest.
Indeed, we provide in \Cref{sec:l_infty-spanners} an application of our sparse neighborhood cover
to fast spanner construction (with near-optimal stretch-size tradeoff)
in $\ell_\infty^d$ for $d=n^\epsilon$. 

\paragraph{$\ell_2$ metrics.}
Our result for $\ell_2$ is an almost-linear-time algorithm that achieves near-optimal quality, namely $\tO(\sqrt{\log n})$. When the algorithm’s parameters are tuned to run in near-linear time, the algorithm essentially amounts to applying the folklore shifted-grid approach.

\begin{theorem}\label{thm:fast-sep-decom}
    Fix $0 < \rho < 1$.
    There is a separating decomposition algorithm for $n$-point metrics in $\ell_2^d$ with quality $O(\sqrt{\rho^{-1} \log n \log \log n})$ and running time $\tilde{O}(n^{1+\rho} + d^2n)$.
    The algorithm is Las Vegas, meaning that with high probability it reports a partition sampled from a separating decomposition, and with the remaining probability it reports ``Fail''.
\end{theorem}

An optimal quality bound $O(\sqrt{\log n})$ was established in~\cite{CCGGP98} 
without an explicit analysis of its algorithmic efficiency.
A naive algorithmic implementation of their proof runs in time $\tilde{O}(n^3 + dn)$,
which can be improved to $\tilde{O}(n^2 + dn)$ using standard methods, as explained in~\cite[Section 2]{Stein_Thesis}.
However, their more sophisticated improvement of the running time to $\tilde{O}(n^{1.51} + dn)$
using Locality Sensitive Hashing (LSH), claimed in~\cite[Section 3]{Stein_Thesis}, 
is unfortunately flawed.

We remark that \Cref{thm:fast-sep-decom} implies also a fast separating-decomposition algorithm 
for $\ell_p$ metrics, $p \in (2,\infty]$,
by just applying known results of~\cite{KP25, KPS25, NR26}.
These results employ a sequence of mappings that essentially reduce 
the problem of constructing a separating decomposition in $\ell_p$ to the same problem in $\ell_2$.
We can simply replace the algorithm of~\cite{CCGGP98} for $\ell_2$ with ours,
and thereby speed up the overall running time.

The proof of \cref{thm:fast-sep-decom} is a surprisingly simple reduction 
to the low-dimensional setting for $\ell_2^{d'}$,
for which we use a separating-decomposition algorithm of Andoni and Indyk~\cite{AI06}. 
In a nutshell, the reduction divides the coordinates 
into blocks of a carefully chosen length and works on each block separately.
We additionally show in \Cref{sec:det-nns} that the same simple idea 
can be applied also to the \emph{nearest-neighbor search} problem (NNS),
yielding new deterministic and space-efficient data structures.

\paragraph{Additional note.}
After completing the current version of the paper (to appear in ESA 2026), 
we obtained, with assistance from \texttt{GPT-5.6-sol Ultra}, 
an improvement to \Cref{thm:fast-sep-decom},
namely, an algorithm that computes a separating decomposition in $\ell_2$
with near-optimal quality $\tO(\sqrt{\log n})$ in near-linear time. 
In comparison, \Cref{thm:fast-sep-decom} 
can achieve such $\tO(\sqrt{\log n})$ quality but in almost-linear time. 
We will include this improved result in a forthcoming revision of the paper.

\section{Padded Decomposition in $\ell_\infty^d$}\label{sec:padded-decom}

This section proves \Cref{thm:fast-padded-decom} 
by presenting our padded-decompositions algorithm for a dataset $X \subseteq \ell_\infty^d$. 
We use $d_X$ to denote the induced metric on a subset $X \subset \ell_\infty^d$. 
The closed ball of radius $r\ge0$ around $x$ in $X$ is denoted 
$B_X(x,r) := \{y \in X : \|x-y\|_\infty \leq r\}$. 
We also abbreviate $B_{\ell_\infty^d}(x,r)$ as $B_{\infty}(x,r)$. 

\subsection{Sparse Neighborhood Covers}

Our algorithm goes through a \emph{sparse neighborhood cover}~\cite{AP90}, 
which is a deterministic counterpart of padded decomposition, defined as follows. 

\begin{definition}[Neighborhood Cover]\label{def:nbr-cover}
    A collection $\cC$ of subsets of $X$ (called \emph{clusters}) is a \emph{$(\beta,s,\Delta)$-neighborhood cover} if the following hold:
    \begin{itemize}
        \item (Cluster Diameter) Each cluster $C \in \cC$ has diameter at most $\Delta$.
        \item (Covering) For each $x \in X$ there exists a cluster $C \in \cC$ such that $B_X (x,\frac{\Delta}{\beta}) \subseteq C$.
        \item (Sparsity) Each $x \in X$ is contained in at most $s$ clusters from $\cC$.
    \end{itemize}
\end{definition}

A common way to construct a neighborhood cover of sparsity $s$ is to take $s$ independent partitions drawn from a padded decomposition.
Recently, Conroy and Filtser~\cite{CF25} showed that the other direction also holds: Given a $(\beta, s, \Delta)$-neighborhood cover, one can generate a $(O(\beta \log s), \Omega(\frac{1}{\beta}), \Delta)$-padded decomposition.
Using their reduction, our main technical contribution is a fast algorithm for sparse neighborhood covers in $\ell_\infty^d$, which, in fact, covers the dataset with ``continuous'' axis-aligned rectangles in $\mathbb{R}^d$.

\begin{definition}
    An axis-aligned rectangle in $\mathbb{R}^d$ is a set of the from $R = I_1 \times \cdots \times I_d$, where each $I_i$ is an interval in $\mathbb{R}$.
    (We allow intervals of all forms $[a,b], (a,b],[a,b),(a,b)$.)
    The \emph{width} of $R$ is the length of the longest interval among $I_1,\dots,I_d$. 
    Note that the intersection of two axis-aligned rectangles is another axis-aligned rectangle (or empty), and every ball $B_\infty (x,r)$ in $\ell_\infty^d$ is an axis-aligned rectangle of width $2r$.
\end{definition}

\begin{theorem}\label{thm:sparse-cover-ell-infty}
    There is an algorithm that, given dataset $X \subseteq \ell_\infty^d$ of size $|X| = n$, diameter bound $\Delta > 0$ and parameter $0 < \rho < 1$, outputs a collection $\cR$ of axis-aligned rectangles with associated cluster collection $\cC(\cR) = \{ C(R) := X \cap R\}_{R \in \cR}$ such that the following properties hold true.
    For $s = O(n^{\rho} \cdot \rho^{-1} \log n)$ and $\beta = O(\rho^{-1} \log \log d)$:
    \begin{itemize}
        \item (Diameter) Each rectangle $R \in \cR$ has width at most $\Delta$.
        \item (Covering) For each $x \in X$, there is covering rectangle $R \in \cR$ 
        such that $B_\infty (x, \frac{\Delta}{\beta}) \subseteq R$.
        \item (Sparsity) Each $x \in X$ is contained in at most $s$ rectangles from $\cR$.
    \end{itemize}
    In particular, the clusters collection $\cC(\cR)$ forms a $(\beta,s,\Delta)$-neighborhood cover of $X$.
    The algorithm is deterministic and runs in time $O(n^{1+\rho} \cdot \rho^{-1} d^2 \log n)$.
\end{theorem}

We now focus on the proof of~\Cref{thm:sparse-cover-ell-infty}; later, in~\Cref{sec:from-snc-to-padded}, we explain how it is used through the reduction of~\cite{CF25} to get the padded decomposition algorithm of~\Cref{thm:fast-padded-decom}.

\paragraph{Terminal reduction.}
The algorithm is based on a \emph{terminal reduction} procedure.
Given a set of terminals $T$ within $X$, this procedure constructs a sparse collection of rectangles which covers a large fraction of the given terminals, leaving the uncovered ones as the new, smaller terminal set.
The formal guarantees are stated in the following lemma:

\begin{lemma}[Terminal Reduction]\label{lem:terminal-reduction}
    There is an algorithm $\Cover(D,X,T,\rho)$, whose input consists of axis-aligned rectangular domain $D = I_1 \times \cdots \times I_d \subseteq \mathbb{R}^d$, dataset $X$ of $n$ points from $D$, terminal set $T \subseteq X$ of points with $\ell_\infty$-distance more than $1$ from the boundary of $D$, diameter bound $\Delta > 0$, and parameter $0 < \rho < 1$.

    It outputs a new terminal set $T^{\out} \subseteq T$ such that $|T^{\out}| \leq |T| - |T|^{1-\rho}$, along with a collection $\cR$ of axis-aligned rectangles with associated clusters $\cC(\cR) = \{X \cap R\}_{R \in \cR}$, such that each rectangle in $\cR$ is contained within the domain $D$, and the following hold:
    \begin{itemize}
        \item (Diameter) Each rectangle $R \in \cR$ has width at most $\Delta$.
        \item (Covering) For each $x \in T \setminus T^{\out}$, there is some $R \in \cR$ such that
        $B_\infty (x ,\tfrac{\Delta}{\beta}) \subseteq R$, where $\beta = O(\frac{\log \log d}{\rho})$.
        \item (Sparsity) The collection $\cR$ can be partitioned into at most $\log |T| + 1$ sub-collections, 
        where each sub-collection consists of pairwise-disjoint rectangles.
    \end{itemize}
    The algorithm is deterministic and runs in $O(n \cdot d^2 \log |T|)$ time.
\end{lemma}

\Cref{thm:sparse-cover-ell-infty} follows by repeated invocations of~\Cref{lem:terminal-reduction}, as follows:

\begin{proof}[Proof of~\Cref{thm:sparse-cover-ell-infty}.]
By scaling the dataset, we may assume that $\Delta = \beta = O(\frac{\log \log d}{\rho})$.
Define $M = \max_{x\in X}\{\|x\|_\infty\}$ and set the domain as $D = [-(M+2), M+2]^d$.
The algorithm works in terminal reduction steps.
Initially, $T_0 = X$.
In step $i$, apply $\Cover(D,X,T_{i-1},\rho)$ of~\Cref{lem:terminal-reduction} to obtain new terminal set $T_i \subseteq T_{i-1}$ and rectangle collection $\cR_i$.
The process halts at the first step $k$ where $T_k = \emptyset$.
Since $|T_i| \leq |T_{i-1}| - |T_{i-1}|^{1-\rho}$ and $|T_0| = n$, we have $k = O(n^{\rho} \cdot \rho^{-1})$ (see~\cite[Claim 3.3]{AP90}).
We let $\cR = \cR_1 \cup \cdots \cup \cR_k$.
\begin{itemize}
    \item (Diameter) \Cref{lem:terminal-reduction} ensures that every rectangle in every $\cR_i$ has width at most $\Delta$.
    \item (Covering) Let $x \in X$, and consider the last step $i$ such that $x \in T_{i-1} \setminus T_i$.
    Then~\Cref{lem:terminal-reduction} implies that there is some $R \in \cR_i$ such that
    $B_\infty (x, \frac{\Delta}{\beta}) = B_\infty (x,1) \subseteq R$.
    \item (Sparsity) \Cref{lem:terminal-reduction} ensures that each point $x \in X$ belongs to $O(\log n)$ rectangles from each $\cR_i$, and hence only to $O(k \log n) = O(n^{\rho} \cdot \rho^{-1} \log n)$ rectangles overall. 
\end{itemize}
Each of the $k = O(n^{\rho} \cdot \rho^{-1})$ invocations of~\Cref{lem:terminal-reduction} takes $O(n \cdot d^2 \log n)$ time, so the total time is $O(n^{1+\rho} \cdot \rho^{-1} d^2 \log n)$.
\end{proof}

\subsection{Proof of~\Cref{lem:terminal-reduction}}

The main remaining part is thus proving~\Cref{lem:terminal-reduction}.
To this end, we utilize powerful tools for $\ell_\infty^d$ that were presented by Indyk in his work on the Nearest-Neighbor Search problem~\cite{Indyk01_ell_infty}.

Informally, Indyk proved that given a finite terminal set in $\ell_\infty^d$, one can either (a) find a dense box (i.e., $\ell_\infty$-ball) of small radius that contains many terminals, or (b) find an axis-aligned hyperplane that has only few terminals close to it, and bisects the rest of the terminals in a roughly balanced way.

\begin{definition}\label[definition]{def:separator}
    Let $T \subseteq \ell_\infty^d$, and $\alpha, \beta, \gamma \in [0,1]$ such that $\alpha + \beta + \gamma = 1$.
    An $(\alpha, \beta ,\gamma)$-separator of $T$ is an axis-aligned hyperplane 
    $H_i (t) = \{x \in \mathbb{R}^d \mid x_i = t\}$
    such that:
    \begin{itemize}
        \item $T_{\Left} := \{x \in T \mid x_i < t-1\}$ is of size $|T_{\Left}| = \alpha |T|$,
        \item $T_{\Mid} := \{x \in T \mid t-1 \leq x_i \leq t+1\}$ is of size $|T_{\Mid}| = \beta |T|$,
        \item $T_{\Right} := \{x \in T \mid t+1 < x_i\}$ is of size $|T_{\Right}| = \gamma |T|$.
    \end{itemize}
\end{definition}

\begin{lemma}[Restatement of~\protect{\cite[Lemma 1]{Indyk01_ell_infty}}]\label{lem:dense-ball-sep}
Given a set $T \subseteq \ell_\infty^d$ and parameter $0 < \rho \leq 1$, there is an $O(d|T|)$-time algorithm whose output is one of the following options:
\begin{enumerate}
    \item (Separator) An $(\alpha,\beta,\gamma)$-separator of $T$ such that $\alpha, \gamma \geq \frac{1}{4d}$ and $\alpha^{1-\rho} + \gamma^{1-\rho} \geq 1$.

    \item (Dense Box) A point $x^* \in T$ such that, for $r = O(\frac{\log \log d}{\rho})$, we have $|B_T(x^*, r)| \geq \frac{|T|}{2}$.
\end{enumerate}
\end{lemma}

\begin{remark}
    In~\cite{Indyk01_ell_infty}, the inequality in the separator case is $\log_{\frac{1}{\beta+\gamma}}  \left( \frac{\beta+ \gamma}{\gamma} \right) \leq \rho$, which is equivalent to $\beta \leq \gamma^{1/(1+\rho)} - \gamma$.
    This implies our inequality:
    $$1=\alpha+\beta+\gamma \leq \alpha+\gamma^{1/(1+\rho)}\leq \alpha+\gamma^{1-\rho} \leq \alpha^{1-\rho}+\gamma^{1-\rho}$$ 
    (where the middle inequality utilizes the fact that $1/(1+\rho) \geq 1-\rho$).
\end{remark}

We are now ready to give the algorithm of~\Cref{lem:terminal-reduction}.
The basic idea is to apply~\Cref{lem:dense-ball-sep} on the given terminal set, and recurse according to its output.
If we get a dense box, we just add it (or, more precisely, its intersection with the domain) to the rectangle cover and recurse on the remaining, uncovered terminals.
If we get a separator, we ``give up'' on covering the few terminals in the middle (namely, include them in $T^{\out}$), and recurse on the left ones and on the right ones.

\begin{mdframed}[style=MyFrame, nobreak=true, leftmargin=-1em,
  rightmargin=-1em,
  innerleftmargin=6pt,
  innerrightmargin=6pt]

\begin{center}
\textbf{Algorithm $\Cover(D, X,T,\rho)$}
\end{center}

\medskip
\noindent\textbf{Input:} rectangular domain $D$, dataset $X \subseteq D$, terminal set $T \subseteq X$, parameter $0 < \rho \leq 1$.

\medskip
\noindent\textbf{Output:} new terminal set $T^{\out}$, rectangle collection $\cR$.

\begin{itemize}\setlength{\itemsep}{2pt}
    \item Base Case: If $T = \emptyset$, return $(\emptyset,\emptyset)$.
    \item Recursive Step: Apply~\Cref{lem:dense-ball-sep} on $T$ with parameter $\rho$.
    \begin{enumerate}\setlength{\itemsep}{2pt}
        \vspace{0.5\baselineskip} %
        \item Case 1: 
        If found $(\alpha,\beta,\gamma)$-separator $H_i (t)$ such that $\alpha, \gamma \geq \frac{1}{4d}$ and $\alpha^{1-\rho} + \gamma^{1-\rho} \geq 1$.
        \begin{itemize}\setlength{\itemsep}{2pt}
            \vspace{0.5\baselineskip} %
            \item Compute the sets $T_{\Left},T_{\Mid},T_{\Right}$ corresponding to $H_i (t)$ (see~\Cref{def:separator}).                        \item Define $D_{\Left} = \{x \in D \mid x_i < t\}$ and $D_{\Right} = \{x \in D \mid x_i > t\}$.

            \item Compute $X_{\Left} = D_{\Left} \cap X$ and $X_{\Right} = D_{\Right} \cap X$.
            \item Let $(T_{\Left}^{\out}, \cR_{\Left}) \gets \Cover(D_{\Left}, X_{\Left}, T_{\Left}, \rho)$ and $(T_{\Right}^{\out}, \cR_{\Right}) \gets \Cover(D_{\Right}, X_{\Right}, T_{\Right}, \rho)$.
            \item Return $(T^{\out} := T_{\Left}^{\out} \cup T_{\Mid} \cup T_{\Right}^{\out}, \cR := \cR_{\Left} \cup \cR_{\Right})$.
        \end{itemize}

        \vspace{0.5\baselineskip} %
        \item Case 2: 
        Else, found $x^* \in T$ such that $|B_T(x^*, r)| \geq \frac{|T|}{2}$ for $r = O(\frac{\log \log d}{\rho})$.
        \begin{itemize}\setlength{\itemsep}{2pt}
            \vspace{0.5\baselineskip} %
            \item Let $(T^{\out}, \cR') \gets \Cover(D, X, T \setminus B_T (x^*, r), \rho)$.
            \item Define $R_{\new} = B_\infty (x^*, r+1) \cap D$, with associated cluster $C(R_{\new}) = X \cap R_{\new}$.
            \item Return 
            $\left( 
                T^{\out}, 
                \cR := \cR' \cup \{R_{\new}\} 
            \right)$.
        \end{itemize}
    \end{enumerate}
\end{itemize}
\end{mdframed}

\paragraph{Correctness.}
We prove the correctness of $\Cover(D,X,T,\rho)$ by induction on $|T|$.
The base case $|T| = 0$ is trivial, so we assume $|T| > 0$ from now on, and denote the output of $\Cover(D,X,T,\rho)$ by $(T^{\out}, \cC)$.
Note that the induction hypothesis applies to all recursive invocations of $\Cover$.
Indeed, in Case 1 we invoke it with terminal sets of size $(1-\alpha)|T|$ and $(1-\gamma)|T|$, and in Case 2 with a terminal set of size $\leq \frac{1}{2} |T|$.
Furthermore, in both cases, the terminals in the recursive call all have $\ell_\infty$ distance of more than $1$ to the boundary of the domain in this call.

\proofsubparagraph*{Diameter.}
    The fact that each rectangle in $\cR$ has width $O(\frac{\log \log d}{\rho})$ is
    immediate by induction: 
    In Case 1 we do not add any new rectangles other than the ones recursively constructed. 
    In Case 2, we add one more rectangle of width $\leq 2(r+1) = O(\frac{\log \log d}{\rho})$.

\proofsubparagraph*{Covering.}
    Let $x \in T \setminus T^{\out}$, and we should show there is some $R \in \cR$ such that $B_\infty (x,1) \subseteq R$.
    
    We start with Case 1, and denote by $H_i (t)$ the $(\alpha,\beta,\gamma)$ separator of this case.
    We have that $T \setminus T^{\out} = T \setminus (T_{\Left}^{\out} \cup T_{\Mid} \cup T_{\Right}^{\out})$, so either $x \in T_{\Left} \setminus T_{\Left}^{\out}$ or $x \in T_{\Right} \setminus T_{\Right}^{\out}$.
    If the former (resp., latter) holds, then by induction, there is $R \in \cR_{\Left}$ (resp. $R \in \cR_{\Right}$) s.t.\ $B_\infty (x,1) \subseteq R$.

    Next, consider Case 2.
    Suppose first that $x \in B_T (x^*, r)$.
    In this subcase, we argue that the newly added rectangle $R_{\new} = B_\infty (x^*, r+1) \cap D$ contains $B_\infty (x,1)$.
    Indeed, since $x \in T$, its $\ell_\infty$-distance from the boundary of $D$ is more than $1$, so this follows by triangle inequality.
    Finally, suppose $x \in (T \setminus B_T (x^*, r)) \setminus T^{\out}$.
    Recall that $T^{\out}$ is the new terminal set returned by the recursive call $\Cover(D, X, T \setminus B_T (x^*, r), \rho)$, so by induction, there exists $R \in \cR'$ that contains $B_\infty (x,1)$.

\proofsubparagraph*{Sparsity.}
    We now show that $\cR$ can be partitioned into at most $\log |T| + 1$ sub-collections, each consisting of pairwise-disjoint rectangles.
    
    First consider Case 1.
    Note that rectangles in $\cR_{\Left}$ (resp., $\cR_{\Right}$) are contained within $D_{\Left}$ (resp., $D_{\Right}$) because they were returned by invoking $\Cover$ with domain $D_{\Left}$ (resp., $D_{\Right}$).
    By induction hypothesis, both $\cR_{\Left}$ and $\cR_{\Right}$ can be partitioned into at most $\log|T| + 1$ sub-collections of pairwise-disjoint rectangles, and we can concatenate the $i$-th sub-collection of $\cR_{\Left}$ with the $i$-the sub-collection of $\cR_{\Right}$ to obtain the $i$-th sub-collection of $\cR$.

    For Case 2, we add one more rectangle $R_{\new}$ to the recursively constructed $\cR'$.
    By induction, the latter can be partitioned into at most $\log \big| T \setminus B_T (x^*,r) \big| + 1 \leq \log \frac{|T|}{2} + 1 = \log |T|$ sub-collections of pairwise-disjoint rectangles.
    So, we simply put $R_{\new}$ in a new sub-collection $\{R_{\new}\}$.

\proofsubparagraph*{Terminal reduction.}
    To conclude the correctness proof, we show that $|T^{\out}| \leq |T| - |T|^{1-\rho}$.

    In Case 1, we have
    \begin{align*}
        |T^{\out}| 
        &= |T_{\Left}^{\out}| + |T_{\Right}^{\out}| \\
        &\leq (|T_{\Left}| - |T_{\Left}|^{1-\rho}) + (|T_{\Right}| - |T_{\Right}|^{1-\rho}) && \text{by induction hypothesis,}\\
        &= (\alpha + \gamma) |T| - (\alpha^{1-\rho} + \gamma^{1-\rho}) |T|^{1-\rho}   && \text{as $|T_{\Left}| = \alpha |T|$, $|T_{\Right}| = \gamma |T|$,}\\
        &\leq |T| - |T|^{1-\rho} && \text{as $\alpha + \gamma \leq 1$, and $\alpha^{1-\rho} + \gamma^{1-\rho} \geq 1$.}
    \end{align*}
    In Case 2, denote $B := B_T (x^*, r)$.
    \begin{align*}
        |T^{\out}| 
        &\leq (|T|-|B|) - (|T|-|B|)^{1-\rho} && \text{by induction hypothesis,} \\
        &\leq |T| - (|B|^{1-\rho} + (|T|-|B|)^{1-\rho} ) && \text{using $|B| \geq |B|^{1-\rho}$,}\\
        &\leq |T| - (|B| + |T| - |B|)^{1-\rho} && \text{by sub-additivity of $f(z) = z^{1-\rho}$} \\
        &= |T| - |T|^{1-\rho}.
    \end{align*}

\paragraph{Running time.}
Finally, we analyze the running time of $\Cover(D,X,T,\rho)$.
Consider the recursion tree of the algorithm, where each node is associated with an instance of dataset $X'$ and terminal set $T'$.
The work done inside this node consists of applying~\Cref{lem:dense-ball-sep}, which takes $O(d|T'|) \leq O(d|X'|)$ time, and spending another $O(d|X'|)$ time to partition $X'$ into $X'_{\Left},X'_{\Right}$ or to compute $B_{T'} (x^*, r)$ and the rectangle $R_{\new}$ with its associated cluster $C(R_{\new})$.
Note that the datasets of instances from a given level $i$ of the recursion tree are disjoint subsets of $X$,
so the total time spent in each level 
$i$ is $O(|X| \cdot d)$.

It thus remains to bound the depth of the recursion tree by $O(d \log |T|)$.
For this, we assert that the instance of a node in depth $i$ has terminal set of size at most $(1-\frac{1}{4d})^i |T|$, which means that the leaves (with terminal sets of size $0$) must have depth $O(d \log |T|)$.
The assertion follows easily by induction on $i$, using the fact that $\alpha,\gamma \geq \frac{1}{4d}$ in Case 1, and that $B_{T} (x^*, r) \geq \frac{|T|}{2} \geq \frac{1}{4d} |T|$ in Case 2.

\medskip
This concludes the proof of~\Cref{lem:terminal-reduction}.

\subsection{From Neighborhood Cover to Padded Decomposition}\label{sec:from-snc-to-padded}

We now explain the details of the reduction from our rectangle-based sparse neighborhood cover to padded decomposition, which gives the proof of~\Cref{thm:fast-padded-decom}.

The reduction is a slight modification of the one by~\cite{CF25} to ensure fast running time.
The main point is the following: 
Given a $(\beta,s,\Delta)$-neighborhood cover $\cC$, the reduction of~\cite{CF25} uses the ``boundary distances'' $\partial_C (x) = \min_{y \in X \setminus {C}} \{d_X (x,y)\}$ defined for every point $x \in X$ and cluster $C \in \cC$ to construct an $(O(\beta \log s), \frac{1}{4\beta}, \Delta)$-padded decomposition.
However, this ``discrete'' definition of $\partial_C$ is problematic for us to compute in less than $O(n^2)$ time.
To overcome this, we utilize the fact that our clustering is in fact induced by intersections with ``continuous'' axis-aligned rectangles in $\mathbb{R}^d$, which allows to redefine the values $\partial_C (x)$ while preserving all their important properties.

From now on, assume~\Cref{thm:sparse-cover-ell-infty} was applied with diameter bound $\Delta$ and parameter $\rho$ to be chosen later, let $\cR$ be the resulting rectangle collection, and $\cC(\cR)$ the resulting cluster collection which forms a $(\beta, s, \Delta)$-neighborhood cover with $\beta=O(\rho^{-1}\log\log d)$ and $s=O(n^\rho \rho^{-1}\log n)$. 
Let $x \in X$ and $R \in \cR$ with corresponding cluster $C(R)$.
We define $\partial_{C(R)} (x)$ as the $\ell_\infty$-distance between $x$ and $\mathbb{R}^d \setminus R$. 
Formally, let $R = I_1 \times \cdots \times I_d$, where the interval $I_i \subseteq \mathbb{R}$ has boundary points $a_i \leq b_i$.
Then,
\begin{equation}\label{eq:boundary-distances}
    \partial_{C(R)} (x) = \inf\{ \|x-y\|_\infty : y \in \mathbb{R}^d \setminus R\}
    =  \max\left\{0,~ \min_{i=1,\dots,d} \big\{\min\{x_i-a_i, b_i - x_i\}\big\} \right\}.
\end{equation}
Note that we can compute all of the nonzero values in $\{\partial_{C(R)} (x) : x\in X, C(R) \in \cC(\cR)\}$ in just $O(n \cdot s \cdot d) = O(n^{1+\rho} \cdot \rho^{-1} d)$ time, since $\partial_{C(R)} (x)$ can be nonzero only if $x$ belongs to the rectangle $R$, and there are at most $s$ such rectangles in $\cR$.
This is the crucial point where we use the fact our neighborhood cover $\cC(\cR)$ is induced by intersections with axis-aligned rectangles to allow fast computation.

We now use these values to define the padded decomposition exactly as in~\cite{CF25}.
Let $\text{Texp}(\lambda)$ be the truncated exponential distribution with parameter $\lambda$, with density function $g(z) = \frac{\lambda e^{-\lambda z}}{1-e^{-\lambda}}$ for $z \in [0,1]$.
For every cluster $C \in \cC(\cR)$, we sample (independently) $\delta_C \sim \text{Texp}(2+2s)$, and define the function $f_C : X \to \mathbb{R}$ by
\[
f_C (x) = \delta_C \cdot \frac{\Delta}{\beta} + \partial_C (x).
\]
Define the partition $\cP = \{P_C \}_{C \in \cC(\cR)}$ where each point $x \in X$ joins the cluster $P_C$ such that $f_C (x)$ is maximized.

First, we show that $\cP$ can be computed in just $O(n \cdot s \cdot d)$ time.
Fix $x \in X$, and let $\cR_x$ be the rectangles in $\cR$ that contain $x$, which are at most $s$.
We claim that $x$ can only join $P_C$ if $C = C(R)$ for some $R \in \cR_x$, so we only need to compute $s$ values of $f_C (x)$ to determine which $P_C$ contains $x$.
This is seen by the observing the following:
\begin{itemize}
    \item By the covering property, some $R \in \cR_x$ contains $B_\infty (x,\frac{\Delta}{\beta})$, so $f_{C(R)} (x) \geq \partial_{C(R)} (x) \geq \frac{\Delta}{\beta}$.
    \item On the other hand, if $x \notin R$, then $\partial_{C(R)} (x) = 0$, and hence $f_{C(R)} (x) < \frac{\Delta}{\beta}$ w.p.\ $1$.
\end{itemize}

Note that the above argument showed that $P_{C(R)} \subseteq C(R) \subseteq R$.
Hence, the diameter of $P_{C(R)}$ is at most the width of $R$, which is at most $\Delta$, so $\cP$ is indeed $\Delta$-bounded.

The proof of the padding property from~\cite{CF25} can be applied verbatim.
This is because the only property of the values $\partial_C (x)$ used by~\cite{CF25} for this proof is the triangle inequality $|\partial_C (x) - \partial_C (y)| \leq d_X (x,y) = \|x-y\|_\infty$, which also holds for our definition in~\Cref{eq:boundary-distances}.

We thus obtained that $\cP$ is indeed a partition sampled from a $(O(\beta \log s), \frac{1}{4\beta}, \Delta)$-padded decomposition.
The running time is dominated by computing the rectangle neighborhood cover in~\Cref{thm:sparse-cover-ell-infty}.
Choosing $\rho=\frac{\log\log n}{\log n}$ completes the proof of~\Cref{thm:fast-padded-decom}.

\subsection{Application: Spanners in $\ell_\infty^d$}\label{sec:l_infty-spanners}

A \emph{spanner with stretch $t \ge 1$} for a finite metric $(X, d_X)$ is a weighted graph $G = (X, E, w)$ with edge weights $w(u, v) = d_X(u, v)$ whose shortest-path distances $d_G$ satisfy $d_X(x, y) \leq d_G(x, y) \leq t \cdot d_X(x, y)$ for all $x, y \in X$. 
The main quality measure for spanners is the tradeoff between the stretch and the \emph{sparsity} (i.e., number of edges $|E|$).
Neighborhood covers are well-known to yield spanner algorithms~\cite{AP90, ABCP98, HIS13}. 
In our case, by applying the algorithm of~\Cref{thm:sparse-cover-ell-infty} with $O(\log D_X)$ different diameter bounds, where $D_X = \frac{\max_{x,y \in X}\|x-y\|_\infty}{\min_{x\neq y \in X}\|x-y\|_\infty}$ is the \emph{aspect ratio} of the metric $(X,d_X)$, we obtain the following. 

\begin{theorem}\label{thm:spanner}
There is a deterministic algorithm that, given dataset $X \subseteq \ell_\infty^d$ of size $|X| = n$ with aspect ratio $D_X$, and parameter $0 < \rho < 1$,
outputs a spanner for $(X,d_X)$ with stretch $O(\rho^{-1}\log\log d)$ and sparsity $O(n^{1+\rho}\rho^{-1}\log n \log D_X)$.
The algorithm runs in $O(n^{1+\rho}\rho^{-1}d^2\log n \log D_X)$ time.
\end{theorem}

The proof is a standard argument, included in~\Cref{sec:spanners-appendix} for completeness.
We remark that, while we mainly focused on neighborhood covers as a tool for padded decompositions,
there is a relaxed notion of such covers with \emph{average sparsity} guarantee which is enough to yield spanners.
In the average sparsity variant, the property that each $x \in X$ is contained in at most $s$ clusters is replaced by $\sum_{C \in \cC} |C| \leq n \cdot s$.
For this relaxed notion, our~\Cref{thm:sparse-cover-ell-infty} can be slightly improved to achieve $s = O(n^{\rho} \log n)$, removing a $\rho^{-1}$ factor, and this improvement propagates to the sparsity and running time of~\Cref{thm:spanner}.

We omit the full details, and only sketch the argument as follows:
Adapt the $\Cover$ algorithm so that instead of ``giving up'' on the middle $T_\Mid$, it includes it in both recursive calls on the left and on the right.
In this case, just a single invocation of $\Cover$ produces all of the clusters $\cC$.
The analysis of the running time and the sparsity is by bounding the depth of the recursion tree and the total size of instances in each of its levels, similarly to~\cite{Indyk01_ell_infty}.

\section{Separating Decomposition in $\ell_2^d$}
\label{sec:sep-decom}

This section proves \cref{thm:fast-sep-decom} 
by providing a separating-decomposition algorithm with quality $\Tilde{O}(\sqrt{\log n})$ 
that runs in almost-linear time.
A central component in the proof is a powerful decomposition technique for the entire $\ell_2^d$ space from \cite[Section 3]{AI06}, 
which runs in exponential time in the dimension of the space, as follows.

\begin{theorem}[\protect{\cite[Section 3]{AI06}}]
\label{thm:exp-sep-decom}
    There is a separating decomposition algorithm for $n$-point metrics in $\ell_2^d$ with quality $O(\sqrt{d})$ and running time $\tilde{O}(n\cdot 2^{d\log d})$.
    The algorithm is Las Vegas, meaning that with high probability it reports a partition sampled from a separating decomposition, and with the remaining probability it reports ``Fail''.
\end{theorem}

We are now ready to prove \cref{thm:fast-sep-decom}.

\begin{proof}[Proof of \cref{thm:fast-sep-decom}]
Assume first that $d\leq O(\log n)$. Essentially, this assumption is valid due to the JL-transform~\cite{JL84}, and we formally justify it at the end of the proof.
Without loss of generality, we may assume that $\Delta = 1$ (by rescaling), that $\rho = \Omega(\frac{1}{\log n})$ (since a smaller $\rho$ does not improve the running time), and that $\rho^{-1}$ is an integer.

Arbitrarily partition the coordinates as $I_1 \cup \cdots \cup I_{\rho^{-1}} = [d]$, where each $I_j$ is of size at most $O(\rho \cdot d)$.
For every $I_j$ and $x \in X$, let $x[I_j]$ denote the restriction of $x$ to the coordinates in $I_j$, and let $X[I_j]=\{x[I_j] \mid x \in X\} \subset \ell_2^{O(\rho \cdot d)}$. 
For each $I_j$, apply~\cref{thm:exp-sep-decom} to sample a partition $\cP_j$ from a $(O(\sqrt{\rho \cdot d}), \sqrt{\rho})$-separating decomposition of $X[I_j]$, which takes $\tO(n \cdot 2^{\rho \cdot d \log (\rho \cdot d)})=\tO(n^{1+\rho \cdot \log\log n})$ time. 
If one of the $\rho^{-1}$ applications of \cref{thm:exp-sep-decom} fails, we halt the algorithm and report ``Fail''. Note that by a union bound, this step succeeds with high probability.
Finally, the partition $\cP$ of $X$ is defined as the common refinement of all $\cP_1,\dots, \cP_{\rho^{-1}}$. That is, two points $x,y \in X$ belong to the same cluster of $\cP$ iff $\cP_j(x[I_j])=\cP_j(y[I_j])$ for all $1 \leq j \leq \rho^{-1}$.

We move to analyze the correctness of the described algorithm.
\proofsubparagraph*{Diameter.}
To show that $\cP$ is $1$-bounded, let $x,y \in X$ be points that were clustered together by $\cP$; namely, for every $1 \leq j \leq \rho^{-1}$, $\cP_j(x[I_j])=\cP_j(y[I_j])$.
Then, since every $\cP_j$ is $\rho$-bounded we have 
\[
\|x-y\|_2^2 = \sum_{j=1}^{d}{\|x_j-y_j\|_2^2} = \sum_{j=1}^{t}{\|x[I_j]-y[I_j]\|_2^2} \leq \rho^{-1} \cdot (\sqrt{\rho})^2 = 1.
\]

\proofsubparagraph*{Separation.}
Let $x,y \in X$. 
Then, we have
\begin{align*}
\Pr[\cP(x) \neq \cP(y)] &\leq \sum _{j=1}^{\rho^{-1}}\Pr[\cP_j(x[I_j]) \neq \cP_j(y[I_j])] && \text{by a union bound,} \\ 
&\leq \sum _{j=1}^{\rho^{-1}} O(\sqrt{\rho \cdot d}) \frac{\|x[I_j]-y[I_j]\|_2}{\sqrt{\rho}}   && \text{by def. of $\cP_j$,}\\ 
&= O(\sqrt{d}) \sum _{j=1}^{\rho^{-1}} \|x[I_j]-y[I_j]\|_2 \\ 
&\leq O(\sqrt{d}) \left(\sqrt{\rho^{-1}}\sqrt{\sum _{j=1}^{\rho^{-1}} \|x[I_j]-y[I_j]\|_2^2}\right) &&\text{by Cauchy-Schwarz,} \\ 
&= O(\sqrt{\rho^{-1} \log n}) \cdot \|x - y\|_2 , 
\end{align*}
as required.

\medskip
To conclude, for the case $d = O(\log n)$, we have shown a separating-decomposition algorithm with quality $O(\sqrt{\rho^{-1} \log n})$ that runs in time $\tilde{O}(\rho^{-1} \cdot n^{1+\rho \log \log n})$.
So, \Cref{thm:fast-sep-decom} for this case follows by rescaling $\rho$.

\proofsubparagraph*{General dimension.}
To resolve the case of general $d$, 
apply a JL transform~\cite{JL84} to embed the dataset $X$ into $\ell_2^{O(\log n)}$, 
such that with probability at least $1-\frac{1}{n^6}$, 
the distance between every two dataset points contracts or expands by at most factor $2$.
Then apply the preceding algorithm with diameter bound $\tilde{\Delta}={\Delta}/{4}$, 
obtaining within $\tilde{O}(n^{1+\rho})$ time
a partition sampled from a separating decomposition with separation parameter $\Tilde{\alpha}=O(\sqrt{\rho^{-1}\log n \log \log n})$, 
which induces a partition $\cP$ of the original dataset $X$.

Outputting the partition $\cP$ would cause some issues, 
even though the failure probability of the JL transform is polynomially small. 
First, when it fails, clusters might not have diameter bounded by $\Delta$, 
whereas the output must be a $\Delta$-bounded partition (with probability $1$).
Second, for close-by points, say within distance $\ll \frac{\Delta}{n^6}$, 
the separation probability is dominated by that failure probability,
which is not proportional to the distance.
To solve these issues, execute the following verification steps before outputting $\cP$ (and if one of them fails, we output ``Fail''):
\begin{enumerate}
    \item Diameter verification: 
    For every cluster $C\in\cP$, pick an arbitrary point $x\in C$, and 
    verify that $\|x-y\|_2 \leq {\Delta}/{2}$ for all $y \in C$.
    This can clearly be implemented in deterministic linear time $O(nd)$. 

    \item Close-pairs verification: 
    This step outputs ``Fail'' or ``Success'', with the following guarantees,
    assuming that $d \leq n$ without loss of generality:%
    \footnote{If $d > n$, apply Gram-Schmidt to reduce to a subspace of dimension $\leq n$ in time $O(n^2 d) \leq O(nd^2)$.}
    \begin{itemize}
        \item If some $x,y \in X$ with $\|x-y\|_2 \leq \frac{\Delta}{n^6}$ are separated by $\cP$, the output is ``Fail''.
        \item If all $x,y \in X$ with $\|x-y\|_2 \leq O(d^{3/2}) \cdot \frac{\Delta}{n^6} \leq O(\frac{\Delta}{n^4})$ 
        satisfy $\cP(x) = \cP(y)$, the output is ``Success''.
        (In the contrapositive: if the output is "Fail" then 
        some $x,y \in X$ with $\|x-y\|_2 \leq O(\frac{\Delta}{n^4})$ are separated by $\cP$.)
    \end{itemize}
    This step essentially scans all the pairs whose distance is at most $\frac{\Delta}{n^6}$,  
    but relaxing this threshold is by factor $O(d^{3/2})$.
    It can be implemented in deterministic $\tilde{O}(d^2 n)$ time 
    by applying hashing techniques of \cite{Chan98} for approximate near-neighbor search; 
    see \cite[Section 4]{Stein_Thesis} for a complete description.
\end{enumerate}

We claim that both verification steps succeed with high probability.
Indeed, the diameter verification succeeds whenever the JL transform succeeds, 
in which case every cluster of $\cP$ has diameter at most 
$2\tilde{\Delta} = {\Delta}/{2}$.
As for the close-pairs verification, when the JL transform succeeds, 
each pair $x,y\in X$ with $\|x-y\|_2 \leq O(\frac{\Delta}{n^4})$ is separated with probability at most 
$\tilde\alpha \cdot \frac{2\|x-y\|_2}{\tilde\Delta} \leq O(1/n^3)$, 
so by a union bound, with high probability no such pair is separated and the verification succeeds.

Finally, we show that when the algorithm reports $\cP$, 
it is indeed a partition sampled from an $(O(\sqrt{\rho^{-1} \log n \log \log n}), \Delta$)-separating decomposition. 
Observe that $\cP$ is $\Delta$-bounded 
because by the first verification step every cluster has radius at most ${\Delta}/{2}$.
To verify the separation probability,
denote by $V$ the event that the verification succeeds, 
and by $J$ the event that the JL transform succeeds.
Now consider $x,y \in X$. 
If $\|x-y\|_2 > \frac{\Delta}{n^6}$, then 
\begin{align*}
    \Pr[\cP(x) \neq \cP(y) \mid V] &= \frac{\Pr[\cP(x) \neq \cP(y) \land V]}{\Pr[V]} \\ &\leq \frac{\Pr[\cP(x) \neq \cP(y)]}{\Pr[V]} \\ &\leq O(1)\cdot \Pr[\cP(x) \neq \cP(y)] && \text{as $V$ happens w.h.p.} 
    \\ &\leq O(1)\cdot \big(\Pr[\cP(x) \neq \cP(y) \mid J] + \Pr[\overline{J}]\big) && \text{by law of total probability,} 
    \\ &\leq O(1)\cdot \big(\tilde{\alpha}\cdot \frac{2\|x-y\|_2}{\tilde\Delta} + \frac{1}{n^6}\big) && \text{as $J$ happens w.h.p.} 
    \\ &\leq O(\tilde{\alpha})\frac{\|x-y\|_2}{\Delta} && \text{as $\|x-y\|_2 > \frac{\Delta}{n^6}$;}
\end{align*}
and if $\|x-y\|_2 \leq \frac{\Delta}{n^6}$ 
then clearly $\Pr[\cP(x) \neq \cP(y) \mid V] = 0$.
This concludes the proof.
\end{proof}

\subsection{Application: Deterministic Nearest-Neighbor Search}
\label{sec:det-nns}

The $c$-approximate nearest-neighbor search ($c$-ANN) problem asks to preprocess an $n$-point dataset $P$ in a metric space $(X, d_X)$ into a data structure that can quickly answer the following queries: 
given $q \in X$, report $x \in P$ satisfying $d_X(q, x) \leq c \cdot \min_{y \in P} d_X (q,y)$.
In this section, we apply the simple ``coordinate division'' trick used in our proof of~\Cref{thm:fast-sep-decom}, together with several ingenious ANN algorithms of Indyk, 
to prove the following. 

\begin{theorem}\label{thm:deterministic-nns}
    For every $0<\epsilon<1$, there is a \emph{deterministic} $O(\epsilon^{-2} \log \log n)$-ANN for $n$-point datasets in $\ell_2^d$ or in $\ell_1^d$, 
    with space $\tilde{O}(\epsilon^{-1} n^{1+\epsilon} \cdot \poly(d))$, query time $\poly(d\log n)$, and preprocessing time $\poly(nd)$.
\end{theorem}

To the best of our knowledge, no known deterministic ANN for $\ell_2^d$ or $\ell_1^d$ when $d \gg \log n$ achieves $o(\log n)$ approximation, $\poly (d\log n)$ query time, and $\tO(\epsilon^{-1}n^{1 + \epsilon}\cdot \poly(d))$ space. 
Furthermore, a recent reduction from \cite{KP26} shows that~\Cref{thm:deterministic-nns} 
can be extended to $\ell_p^d$ for every $1 \leq p < \infty$, at the cost of an additional $\epsilon^{-1}p^{O(1) + \log\log p}$ factor in the approximation.

We now discuss the proof of~\Cref{thm:deterministic-nns}.
By employing a reduction from \cite{Indyk00} (see also \cite[Section 3]{AIR19}), for every $c > 1$, 
one can use $\poly(d\log n)$ deterministic $c$-ANNs for \emph{Hamming space} $(\{0,1\}^{O(\log n)}, d_\cH)$ 
to compute (in polynomial time) a deterministic $O(c)$-ANN in $\ell_2^d$ or $\ell_1^d$, with at most a $\poly(d\log n)$ overhead in query time and space.
Thus~\Cref{thm:deterministic-nns} reduces to proving the following lemma.

 \begin{lemma}\label{lem:det-NNS}
    For every $0<\epsilon<1$, there is a \emph{deterministic} $O(\epsilon^{-2}\log\log n)$-ANN for $n$-point datasets in Hamming space $(\{0,1\}^{O(\log n)},d_\cH)$ with query time $O(\log n)$, space $\tO(\epsilon^{-1}n^{1+\epsilon})$ and preprocessing time $\poly(n)$.
\end{lemma}
\begin{proof}
    Let $P \subseteq \{0,1\}^{O(\log n)}$ be the given dataset.
    Divide the $O(\log n)$ coordinates of the space into $r = O(\epsilon^{-1})$ sets $I_1, \dots, I_r$, each of size at most $\epsilon \log n$.
    Let $d_{\product}$ denote the \emph{product metric} over $\{0,1\}^{O(\log n)} = \{0,1\}^{I_1} \times \cdots \times \{0,1\}^{I_r}$, defined by 
    \[
    d_{\product} (x,y) := \max_{j=1,\dots,r} d_\cH (x[I_j], y[I_j]).
    \]
    That is, $d_{\product}(x,y)$ is obtained by considering the Hamming distance restricted to each coordinate block $I_j$, and taking the maximum.
    Clearly, $\frac{1}{r} \cdot d_\cH (x,y) \leq d_{\product} (x,y) \leq d_\cH (x,y)$,
    i.e., the identity mapping from $(\{0,1\}^{O(\log n)}, d_\cH)$ to $(\{0,1\}^{O(\log n)}, d_{\product})$ has distortion at most $r = O(\epsilon^{-1})$.

    For each space $(\{0,1\}^{I_j}, d_\cH)$, construct a trivial $1$-ANN data structure of space $O(2^{|I_j|}) = O(n^\epsilon)$ and query time $O(1)$ that explicitly stores the answers to all possible queries. %
    Furthermore, this $1$-ANN also works for the priority version of the problem,
    where each point in the dataset has a priority, and if there is more than one nearest neighbor, the one with smallest priority should be reported.
    
    We now appeal to a known ANN construction for product metrics, from~\cite[Theorem 1]{Indyk02}, 
    and its slight space refinement in~\cite[Appendix A]{ANNRW17} (which require the aforementioned priority version).
    These imply an $O(\epsilon^{-1} \log \log n)$-ANN for the dataset $P$, considered in $(\{0,1\}^{O(\log n)}, d_{\product})$, with space $\tilde{O}(\epsilon^{-1} n^{1+\epsilon})$ and query time $\polylog n$.
    We use this ANN also for the original Hamming metric:
    Given a query $q$, the returned point $x \in P$ satisfies
    \[
    d_\cH (q,x) \leq r \cdot d_{\product}(q,x) \leq r \cdot O(\epsilon^{-1} \log \log n) \cdot \min_{y \in P} d_{\product}(q,y) \leq O(\epsilon^{-2} \log \log n) \cdot \min_{y \in P} d_\cH(q,y) ,
    \]
    as needed.
\end{proof}

{\small
  \bibliographystyle{alphaurl}
  \bibliography{references}
} %

\appendix

\section{Spanner Construction in $\ell_\infty$}\label{sec:spanners-appendix}

\begin{proof}[Proof of~\Cref{thm:spanner}]
Without loss of generality, assume $\min_{x,y \in X}\|x-y\|_\infty=1$ and $\max_{x,y \in X}\|x-y\|_\infty=D_X$, and let $\beta = O(\rho^{-1}\log\log d)$. 
For every $i \in \{1,\dots, \log D_X + 1 \}$, apply~\Cref{thm:sparse-cover-ell-infty} with $\rho$ and $\Delta_i=2^i \beta$ to obtain a $(\beta, s, \Delta_i)$-neighborhood cover $\cC_i$, where $s=O(n^\rho \rho^{-1} \log n)$.
We construct the spanner by selecting an arbitrary point $s_C$ from every cluster $C \in \cC_i$ and adding the edges $\big\{\{x, s_C\} \mid x \in C \big\}$. 

The spanner's sparsity follows immediately from that of the covers $\cC_1,\dots,\cC_{O(\log D_X)}$: for each $i$, we add a linear number of edges per cluster $C\in\cC_i$, and each point belongs to at most $s=O(n^\rho \rho^{-1} \log n)$ clusters.
To verify the stretch, let $x,y \in X$ and let $i$ such that $\|x-y\|_\infty \in [2^{i-1}, 2^i)$. Let $C\in \cC_i$ be the cluster covering $x$, meaning $B(x,2^i) \subseteq C$.
Since $\|x-y\|_\infty < 2^i$, $y \in B(x,2^i) \subseteq C$, so the spanner contains the path $x \to s_C \to y$. The diameter of $C$ is at most $2^i\beta$, hence, the weight of this path is at most $2^{i+1}\beta \leq 4\beta \|x-y\|_\infty = O(\rho^{-1}\log\log d)\|x-y\|_\infty$, concluding the proof.
\end{proof}

\end{document}